\documentclass[twoside,leqno,twocolumn]{article}

\usepackage[letterpaper]{geometry}

\usepackage{ltexpprt}
\usepackage{hyperref}

\usepackage{cite}
\usepackage{amsmath,amssymb,amsfonts}
\usepackage{graphicx}
\usepackage{textcomp}
\usepackage{xcolor}

\usepackage{algorithm}
\usepackage[noend]{algpseudocode}

\newcommand*\Input[1]{\Statex \textbf{Input:} #1}
\newcommand*\Output[1]{\Statex \textbf{Output:} #1}
\algrenewcommand\alglinenumber[1]{#1}

\usepackage{amsthm}

\newtheoremstyle{def_style}
  {}          
  {}          
  {}          
  {}          
  {\bfseries} 
  {.}         
  {.5em}      
  {}          

\theoremstyle{def_style}
\newtheorem{define}{Definition}

\theoremstyle{def_style}
\newtheorem*{prob}{Problem Statement}

\theoremstyle{def_style}
\newtheorem{guarantee}{Guarantee}

\theoremstyle{def_style}
\newtheorem{lemma}{Lemma}

\theoremstyle{def_style}
\newtheorem{proposition}{Proposition}

\DeclareMathOperator*{\argmin}{arg\,min}

\let\emptyset\varnothing

\begin{document}


\title{\Large Error-bounded Approximate Time Series Joins Using\\Compact Dictionary Representations of Time Series}

\date{}

\author{Chin-Chia Michael Yeh\thanks{Visa Research.} \and Yan Zheng$^*$ \and Junpeng Wang$^*$ \and Huiyuan Chen$^*$ \and Zhongfang Zhuang$^*$ \and Wei Zhang$^*$ \and Eamonn Keogh\thanks{University of California, Riverside.}}

\maketitle

\fancyfoot[R]{\scriptsize{Copyright \textcopyright\ 2022 by SIAM\\
Unauthorized reproduction of this article is prohibited}}

\begin{abstract} \small\baselineskip=9pt
The \textit{matrix profile} is an effective data mining tool that provides similarity join functionality for time series data.
Users of the matrix profile can either join a time series with itself using \textit{intra}-similarity join (i.e., \textit{self-join}) or join a time series with another time series using \textit{inter}-similarity join.
By invoking either or both types of joins, the matrix profile can help users discover both conserved and anomalous structures in the data.
Since the introduction of the matrix profile five years ago, multiple efforts have been made to speed up the computation with approximate joins; however, the majority of these efforts only focus on self-joins.
In this work, we show that it is possible to efficiently perform approximate inter-time series similarity joins with error bounded guarantees by creating a compact ``dictionary" representation of time series.
Using the dictionary representation instead of the original time series, we are able to improve the throughput of an anomaly mining system by at least 20X, with essentially no decrease in accuracy.
As a side effect, the dictionaries also summarize the time series in a semantically meaningful way and can provide intuitive and  actionable insights.
We demonstrate the utility of our dictionary-based inter-time series similarity joins on domains as diverse as medicine and transportation.
\end{abstract}


\section{Introduction}
The \textit{matrix profile} is a versatile and efficient time series data mining tool that helps to solve a variety of problems such as time series motif/discord discovery~\cite{yeh2016matrix}, time series semantic segmentation~\cite{gharghabi2017matrix}, and time series shapelet discovery~\cite{zhu2020swiss}.
It addresses the aforementioned problems by performing either or both types of joins defined for it: \textit{intra}-time series similarity join and \textit{inter}-time series similarity join~\cite{yeh2018time}.
The intra-similarity join (commonly called the \textit{self-join}) captures the nearest neighbor relationship among subsequences within a given time series; inter-similarity join captures the nearest neighbor relationship from one set of subsequences to another set, where each subsequence set comes from different time series.
While there are multiple scalable methods proposed for approximated intra-similarity joins~\cite{zhu2018matrix,zimmerman2019matrix}, to date there is surprisingly little work focused on the approximated inter-similarity join.

In this work, we propose a fast approximated inter-time series similarity join method with guaranteed error bounds.
The proposed method has two phases: a learning phase and an inference phase.
During the learning phrase, we build a compact dictionary to capture local patterns within the time series.
Our ability to create compact, yet faithful dictionaries for a dataset exploits the fact that many time series datasets contain significant redundancies.
Indeed, as hinted at in Fig.~\ref{fig:compress}, the \textit{inverse} of the area under the curve of a matrix profile (see Definition~\ref{def:matrix_pro}) can be seen as a measure of this redundancy.
In the inference phase, we can efficiently join any given time series with the dictionaries  instead of the original database.
As the dictionary is much smaller than the original time series, the similarity join can be performed in a fraction of time compared to joining with the original data.

\begin{figure}[htbp]
\centerline{
\includegraphics[width=0.99\linewidth]{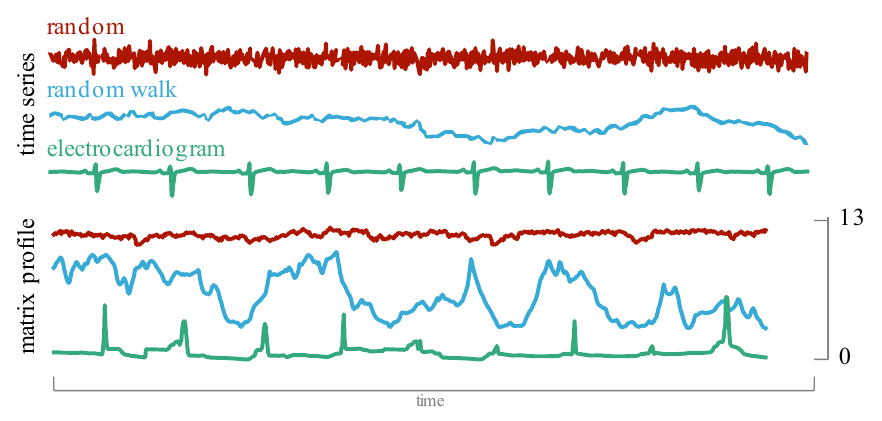}
}
\caption{
\footnotesize
\textit{top-panel}) top to bottom, a random time series, a random walk time series, and ten cardiac cycles of a patient from MIT-BIH Long-Term ECG Database~\cite{goldberger2000physiobank}.
\textit{bottom-panel}) The matrix profiles of the three time series.
The inverse (negative) of the area under the curve for each matrix profile can be seen as a measure of ``compressibility" of each time series.
Note that the three matrix profiles are shown with a commensurate scale to emphasize that it really is the case that the closest pair (the top-1 motif) from the random time series is further apart than the furthest pair (the top-1 discord) from either the random walk or cardiac cycle time series.
}
\label{fig:compress}
\end{figure}

While a faster inter-similarity join will be beneficial whenever joins are performed for any time series data mining tasks~\cite{zhu2020swiss}, here we consider a particular use case where the proposed method plays the central role in improving the performance of the data mining system.
The matrix profile has been shown to allow a simple yet effective time series anomaly detection method~\cite{wu2020current,anton2018time}.
In particular, the highest point in a matrix profile is called the \textit{time series discord}.
By definition, discords are not similar to any other subsequence, and as such often correspond to anomalies~\cite{chandola2009anomaly}.
Virtually all work demonstrating the utility of discords as anomaly detectors consider only the batch case, our work allows the possibility of the online discord discovery in fast arriving streams.
Specifically, we envision labeling an archived dataset as ``normal", and as new data arrives, comparing it to the normal data, flagging subsequences that cannot find a sufficiently close nearest neighbor.

For each arriving query subsequence to the system during detection, the system needs to join the query with the entire normal dataset.
This may be tenable when we have a small normal dataset and a low throughput, but for many real-world applications it is simply untenable.
If we use our proposed approximated inter-similarity join instead, we could reduce the inference time to precisely match the available computational resources.
As we will demonstrate, this allows us to consider sampling rates that are up to 20 times faster, with no appreciable reduction in accuracy.
While our main goal in this work is to provide fast inter-similarity joins, the proposed dictionary representation also brings the benefit of summarizing a time series.
As we will show in Section~\ref{sec:summary}, these summarizations can provide intuitive and actionable insights in datasets where manual inspection by scrolling and zooming interactions would be untenable.

\section{Background and Related Work}
\label{sec:background}
In this section, we begin by introducing the definitions and notation, then we discuss the related work.

\subsection{Definitions and Notation.}
\label{sec:definition}
We begin by defining the data type of interest, \textit{time series}:

\begin{define}
    A \textit{time series} $T \in \mathbb{R}^{n}$ is a sequence of real valued numbers $t_i \in \mathbb{R}:T=\left[t_1, t_2, ..., t_n \right]$ where $n$ is the length of $T$.
\end{define}

For time series data mining tasks, we usually are not interested in \textit{global}, but \textit{local} properties of a time series.
Thus the time series similarity join is performed on local \textit{subsequences} of time series~\cite{yeh2018time}.

\begin{define}
    A \textit{subsequence} $T_{i,m} \in \mathbb{R}^{m}$ of a time series~$T$ is a length~$m$ contiguous subarray of $T$ starting from position~$i$.
    Formally, $T_{i,m}=[t_{i}, t_{i+1}, ..., t_{i+m-1}]$.
\end{define}

We are interested in a similarity join of all subsequences of a given time series.
We define an \textit{all-subsequences set} of a given time series as a set that contains all possible subsequences from the time series.

\begin{define}
    An \textit{all-subsequences set}~$\mathbf{A}_m$ of a time series~$T$ is an ordered set of all possible subsequences of $T$ obtained by sliding a window of length~$m$ across $T: \mathbf{A}_m =\{T_{1,m}, T_{2,m},..., T_{n-m+1,m}\}$, where $m$ is a user-defined subsequence length.
    We use $\mathbf{A}_m[i]$ to denote $T_{i,m}$.
\end{define}

We can take any time series of length~$m$ and compute its distance with each subsequence~$\mathbf{A}_m[i] \in \mathbf{A}_m$.
We use the \textit{distance profile} to store the result distances.

\begin{define}
    A \textit{distance profile}~$S$ is a vector of the distances (or similarities) between a given query $Q \in \mathbb{R}^{m}$ and each subsequence~$\mathbf{A}_m[i]$ in an all-subsequences set~$\mathbf{A}_m$ of $T$.
\end{define}

The most common \textit{distance} function used for distance profile computation is the $z$-normalized Euclidean distance~\cite{mass}.
Some communities (especially seismology~\cite{senobari2019super}) prefer to work with bounded similarity functions and use the Pearson correlation coefficient instead~\cite{mass}.
The distance profile can be considered as a \textit{meta} time series that annotates the time series $T$ that was used to generate all-subsequences set~$\mathbf{A}_m$.
Distance profiles are frequently used in conjunction with the time series similarity join when analyzing time series data~\cite{zhu2020swiss}.
An example of distance profile is illustrated in Fig.~\ref{fig:distance_pro}.

\definecolor{query}{HTML}{339966}
\definecolor{time_series}{HTML}{990000}
\definecolor{distance_pro}{HTML}{0099CC}
\definecolor{similar}{HTML}{FFCC00}
\definecolor{less_similar}{HTML}{663366}

\begin{figure}[htbp]
\centerline{
\includegraphics[width=0.99\linewidth]{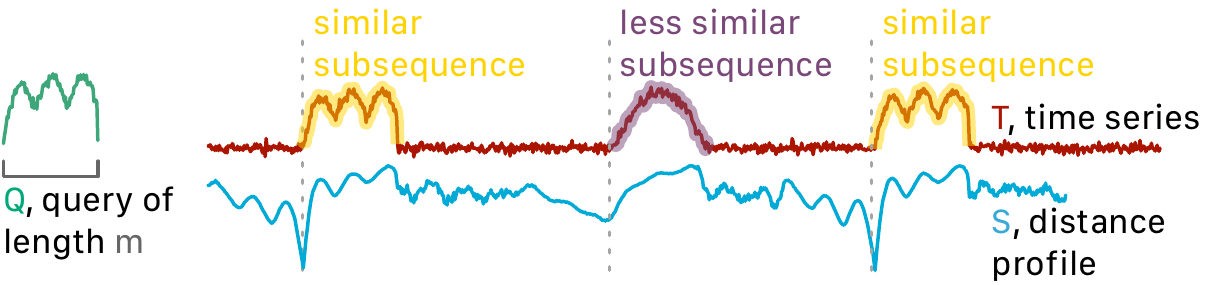}
}
\caption{
\footnotesize
A shorter time series~\textcolor{query}{$Q$} queries the time series~\textcolor{time_series}{$T$} by computing the distance with every subsequence in~\textcolor{time_series}{T}.
The result distances are stored in a vector~\textcolor{distance_pro}{$S$} called \textit{distance profile}.
Subsequences with distinct shape features are highlighted.
The subsequences more similar to the query are highlighted with \textcolor{similar}{yellow}, and the less similar one is highlighted with~\textcolor{less_similar}{purple}.
}
\label{fig:distance_pro}
\end{figure}

Because time series similarity joins concern the nearest neighbor (\textit{1NN}) relation between subsequences based on the distance (or similarity) between them; we define a \textit{1NN}-join function which indicates the nearest neighbor relation between the two input subsequences.

\begin{define}
    Given two all-subsequences sets~$\mathbf{A}_m$ and~$\mathbf{B}_m$ and two subsequences $\mathbf{A}_m[i]$ and $\mathbf{B}_m[j]$, a \textit{1NN-join function} $\theta_{\textit{1NN}}(\mathbf{A}_m[i], \mathbf{B}_m[j])$ is a Boolean function which returns $\texttt{True}$ only if $\mathbf{B}_m[j]$ is the nearest neighbor of $\mathbf{A}_m[i]$ in the set~$\mathbf{B}_m$.
\end{define}

With the defined join function, a \textit{similarity join set} can be generated by applying the similarity join operator on two input all-subsequences sets.

\begin{define}
Given all-subsequences sets $\mathbf{A}_m$ and $\mathbf{B}_m$ of time series~$T_A$ and $T_B$, a similarity join set $\mathbf{J}_{AB,m}$ of $\mathbf{A}_m$ and $\mathbf{B}_m$ is a set containing pairs of each subsequence in $\mathbf{A}_m$ with its nearest neighbor in $\mathbf{B}_m: \mathbf{J}_{AB,m}=\{\langle \mathbf{A}_m[i], \mathbf{B}_m[j] \rangle |\theta_{\textit{1NN}}(\mathbf{A}_m[i], \mathbf{B}_m[j])\}$.
\end{define}

We measure the distance (or similarity) between each pair within a similarity join set and store the result into a vector, the \textit{matrix profile}.

\begin{define}
\label{def:matrix_pro}
Given time series~$T_A$, $T_B$, and subsequence length~$m$, a \textit{matrix profile}~$P_{AB}$ is a vector of the distances (or similarities) between each pair in $\mathbf{J}_{AB,m}$.
We formally denote this operation as $P_{AB} = T_A\bowtie_{\theta_{\textit{1NN}}, m}T_B$.
\end{define}

\definecolor{time_series_a}{HTML}{990000}
\definecolor{time_series_b}{HTML}{FFCC00}
\definecolor{discord}{HTML}{663366}

\begin{figure}[htbp]
\centerline{
\includegraphics[width=0.99\linewidth]{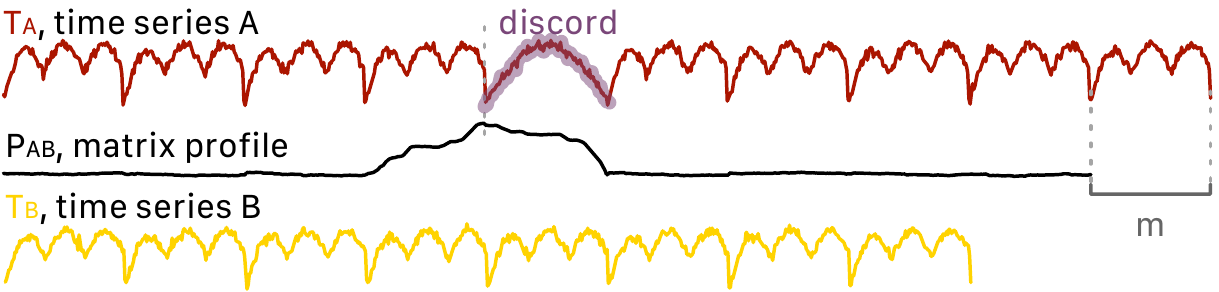}
}
\caption{
\footnotesize
$P_{AB}$ is the matrix profile output by $\textcolor{time_series_a}{T_A}\bowtie_{\theta_{\textit{1NN}}, m}\textcolor{time_series_b}{T_B}$.
$P_{AB}$ reveals the \textcolor{discord}{discord}, i.e., a unique shape feature to \textcolor{time_series_a}{$T_A$} that does not occur in \textcolor{time_series_b}{$T_B$}.
}
\label{fig:matrix_pro}
\end{figure}

Similar to distance profile, the matrix profile can also be considered as a \textit{meta} time series that annotates a time series.
For example, if the all-subsequences set $\mathbf{A}_m$ is extracted from time series $T_A$ and the all-subsequences set $\mathbf{B}_m$ is extracted from time series $T_B$, because we are finding the nearest neighbor for each subsequence in $T_A$, the result matrix profile annotates $T_A$ with each subsequence's nearest neighbor distance in $T_B$.
The matrix profile~$P_{AB}$ reveals the similarities and differences from $T_A$'s perspective as shown in Fig.~\ref{fig:matrix_pro}.

Since $T_A$ and $T_B$ can either be the same time series (i.e., \textit{intra}-time series similarity join) or two different time series (i.e., \textit{inter}-time series similarity join), there are two types of time series similarity join algorithms.
The computation of intra-similarity join matrix profile (either exact or approximate) has already been extensively studied in prior works~\cite{zimmerman2019matrix,zhu2018matrix,zimmerman2019matrixxiv}.
Here we focus on the problem of computing approximate inter-similarity join matrix profile.

\begin{prob}
Given a time series~$T_B$, we want to find a compact dictionary representation~$\mathbf{D}_B$ for~$T_B$ such that $||P_{AB} - \hat{P}_{AB}||$ is minimized for any given time series~$T_A$ where $P_{AB} = T_A\bowtie_{\theta_{\textit{1NN}}, m}T_B$, $\hat{P}_{AB} = T_A\bowtie_{\theta_{\textit{1NN}}, m}\mathbf{D}_B$, and $||\cdot||$ denotes vector norm.
\end{prob}

In addition to optimizing the measure of success defined in the problem statement, we also wish to ensure our approximate algorithm provides the following two performance guarantees: \textit{no false negative guarantee} and \textit{max error guarantee}.
For a discord-based anomaly detection system~\cite{wu2020current}, it is important that the approximate solution captures all potential time series discords (i.e., anomalies).
In other words, the detection system needs to guarantee a recall of $100\%$ (i.e., no false negatives).
To achieve this, the approximated nearest neighbor distance for any subsequence in $T_A$ cannot be smaller than the true nearest neighbor distance.

\begin{guarantee}
\label{gua:false}
Given that $P_{AB}$ and $\hat{P}_{AB}$ are defined using distance measure, the \textit{no false negative guarantee} requires $\hat{P}_{AB}[i] \geq P_{AB}[i]$ for all $i$ and for any given~$T_A$.
\end{guarantee}

Note, if $P_{AB}$ and $\hat{P}_{AB}$ are defined using similarity measure like Pearson correlation coefficient, $\hat{P}_{AB}[i] \leq P_{AB}[i]$ for all~$i$.

This requirement is essentially the same as the lower bounding lemma of the GEMINI~\cite{agrawal1993efficient}.
Satisfying this property allows many downstream higher-level algorithms to reason with the approximate results, yet produces \textit{exact} answers with respect to the original data.

Another desired guarantee for any approximate algorithm is an error-bound to the approximation.

\begin{guarantee}
\label{gua:error}
Given $T_B$ and $\mathbf{D}_B$, the \textit{max error guarantee} requires the existence of an error bound $e_{\text{max}} \in \mathbb{R}$ such that $|\hat{P}_{AB}[i] - P_{AB}[i]| \leq e_{\text{max}}$ for all~$i$ and for any given~$T_A$.
\end{guarantee}

Such error-bound can either be used in a stopping criterion if the dictionary is built greedily or provide confidence to the approximate solution for downstream tasks.
Let us consider time series discord\footnote{Time series discord is defined as the subsequence with the largest matrix profile value within a given time series~\cite{yeh2016matrix}.} discovery as an example.
If the largest approximate matrix profile value~$\hat{P}_{AB}[i]$ and the second largest approximate matrix profile value~$\hat{P}_{AB}[j]$ have a difference greater than~$e_{\text{max}}$, then the subsequence associated with~$\hat{P}_{AB}[i]$ is guaranteed to be the time series discord of~$T_A$.
To demonstrate why the statement is true, given Guarantee~\ref{gua:false}, we define the error for~$\hat{P}_{AB}[i]$ and~$\hat{P}_{AB}[j]$ as:

\vspace{-0.5em}
\begin{equation}
\begin{aligned}
    e_i &\triangleq \hat{P}_{AB}[i] - P_{AB}[i] \\
    e_j &\triangleq \hat{P}_{AB}[j] - P_{AB}[j]
\end{aligned}
\label{eq:error1}
\end{equation}
\vspace{-0.5em}

With algebraic manipulation, Eq.~\ref{eq:error1} becomes:

\vspace{-0.5em}
\begin{equation}
\begin{aligned}
    P_{AB}[i] &= \hat{P}_{AB}[i] - e_i \\
    P_{AB}[j] &= \hat{P}_{AB}[j] - e_j
\end{aligned}
\label{eq:error2}
\end{equation}
\vspace{-0.5em}

Because of Guarantee~\ref{gua:error}, both $e_i$ and $e_j$ are non-negative numbers bounded above by~$e_{\text{max}}$.
Given that $\hat{P}_{AB}[i] - \hat{P}_{AB}[j] > e_{\text{max}}$, $P_{AB}[j]$ can only be greater than~$P_{AB}[i]$ if $e_i > e_{\text{max}}$ even if $e_i$ is $0$.
However, since $e_i > e_{\text{max}}$ contradicts Guarantee~\ref{gua:error}, $P_{AB}[i]$ must be the time series discord of~$T_A$.

\subsection{Related Work.}
The matrix profile has attracted significant attention since its introduction in 2016~\cite{yeh2016matrix}.
There are now dozens of works showing how it can be employed in diverse domains.
For example, a recent paper by a NASA team applied it to ``\textit{observations of the magnetosphere collected by the Cassini spacecraft in orbit around Saturn}", noting that ``the best-performing method was the Matrix Profile"~\cite{daigavane2020detection}.
Given its track record of successful deployments, there are a number of papers on accelerating the matrix profile, either by algorithmic insights~\cite{zhu2018matrix,zimmerman2019matrix} or the use of high-performance hardware~\cite{zimmerman2019matrixxiv}.
Thus far, most of these optimizations have focused only on the self-join case.

Moreover, most of these efforts are focused on the matrix profiles ability to discover \textit{motifs}.
However, as we noted in Guarantee~\ref{gua:error} above, if our main focus is on time series \textit{discords}, under some circumstances we can exploit the \textit{approximate} results from our dictionary join to produce an overall \textit{exact} answer to discord questions.

The main use of time series discords is as an intuitive anomaly detector.
The recent explosion in research on anomaly detection makes it a difficult area to survey.
However, recent works including~\cite{anton2018time, wu2020current,daigavane2020detection,nakamura2020merlin} suggest that time series discords are at least competitive with other state-of-the-art techniques.
In a sense this is unsurprising, anomaly detection typically suffers from a paucity of labeled data (as little as none), and anomaly detectors with more than ten parameters have been proposed~\cite{wu2020current}.
With this imbalance between data availability and representational expressiveness, overfitting is very hard to prevent.
In contrast, time series discords require only one intuitive parameter, the subsequence length.

Finally, the idea of dictionary learning (or more generally \textit{numerosity reduction}) has been explored for strings, graphs, images etc.
In addition, there have been some applications to atomic time series (i.e. individually extracted heartbeats or gait cycles)~\cite{yazdi2018time}, and methods~\cite{marascu2014tristan,khelifati2019corad} that do not provide guarantees.
However, we are not aware of any research that considers long, unsegmented time series, and/or which offer error guarantees.

\section{Algorithm}
\label{sec:algorithm}
The proposed approximate inter-time series similarity join algorithm has two stages: 1) dictionary learning and 2) similarity join.
The dictionary learning algorithm constructs a compact dictionary representation from time series in the first stage; then, any incoming time series subsequences can efficiently join with the dictionary instead of the original time series in the second stage.
We introduce the dictionary learning algorithm and the join process in Section~\ref{sec:dictionary} and Section~\ref{sec:similarity}, respectively.
As we established in Section~\ref{sec:definition}, it is important that the proposed method provides the performance guarantees.
We show the proposed method provides the necessary assurances in Section~\ref{sec:guarantee}.

\subsection{Dictionary Learning.}
\label{sec:dictionary}
Our dictionary learning algorithm is designed based on two observations: 1) time series motifs tend to be the most representative subsequences and 2) subsequences that are similar to the ones already added to the dictionary should be avoided, as they only provide marginal gains in reducing the overall error.
Algorithm~\ref{alg:dictionary} shows the pseudocode of the proposed dictionary learning algorithm.
The inputs to the algorithm are time series~$T_B$ and subsequence length~$m$.

\begin{algorithm}[ht]
    \centering
    \caption{Dictionary Learning\label{alg:dictionary}}
    \footnotesize
    \begin{algorithmic}[1]
        \Input{time series~$T_B$, subsequence length~$m$}
        \Output{dictionary~$\mathbf{D}_B$}
        \Function{LearnDictionary}{$T_B, m$}
        \State $n \gets \textsc{ GetArraySize}(T_B)$
        \State $P_B \gets T_B\bowtie_{\theta_{\textit{1NN}}, m}T_B$
        \State $S \gets \textsc{ GetZeroArrayOfSize}(n - m + 1)$
        \State $\mathbf{D}_B \gets \emptyset$
        \While{\texttt{True}}
        \State $P'_{B} \gets P_B - S$
        \For{\textbf{each} $D_B$ \textbf{in} $\mathbf{D}_B$}
        \State $i \gets \textsc{ GetStartIndex}(D_B)$
        \State $P'_{B}[i - \frac{m}{2}:i + \frac{m}{2}] \gets \infty$
        \EndFor
        \State $j \gets \argmin{P'_{B}}$
        \State $\mathbf{D}_B \gets T_B[j:j + m]$ \Comment{adding new element to $\mathbf{D}_B$}
        \If{\texttt{terminal condition} \textbf{is} \texttt{True}}
        \State \textbf{break}
        \EndIf
        \State $S' \gets \textsc{ GetDistanceProfile}(T_B, T_B[j: j + m])$
        \If{$S$ \textbf{is} \texttt{zero vector}}
        \State $S \gets S'$
        \Else
        \State $S \gets \textsc{ ElementwiseMin}(S, S')$
        \EndIf
        \EndWhile
        \State \Return $\mathbf{D}_B$
        \EndFunction
    \end{algorithmic}
\end{algorithm}

In line 2, the time series length~$n$ is inferred from the time series~$T_B$.
In line 3, the intra-similarity join matrix profile~$P_B$ is extracted from~$T_B$ with subsequence length~$m$.
The particular algorithm we used for the similarity join is based on~\cite{zimmerman2019matrixxiv} in our implementation, and it is possible to use other matrix profile implementations such as \texttt{SCRIMP++}~\cite{zhu2018matrix} for an even faster approximate solution.
Next, in line 4 and line 5, a temporary variable storing the merged distance profile~$S$ and the dictionary~$\mathbf{D}_B$ are initialized as a vector of zeros and an empty set, respectively.
The merged distance profile~$S$ is used to store the inter-similarity join between $\mathbf{D}_B$ and $T_B$.

The main loop of the algorithm starts from line 6 through line 19.
In line 7, the intra-similarity join matrix profile is further processed to highlight the subsequences less similar to the dictionary elements.
Because~$S$ stores the inter-similarity join between $\mathbf{D}_B$ and $T_B$, if subsequence~$i$ is less similar to the dictionary elements comparing to subsequence~$j$ (i.e., $S[i] > S[j]$), $P_B-S$ reduces the matrix profile value associated with subsequence~$i$ more comparing to subsequence~$j$.
We refer to the output of line 7 as the processed matrix profile~$P'_B$.

\definecolor{query}{HTML}{339966}
\definecolor{time_series}{HTML}{990000}
\definecolor{distance_pro}{HTML}{0099CC}

\begin{figure}[htbp]
\centerline{
\includegraphics[width=0.99\linewidth]{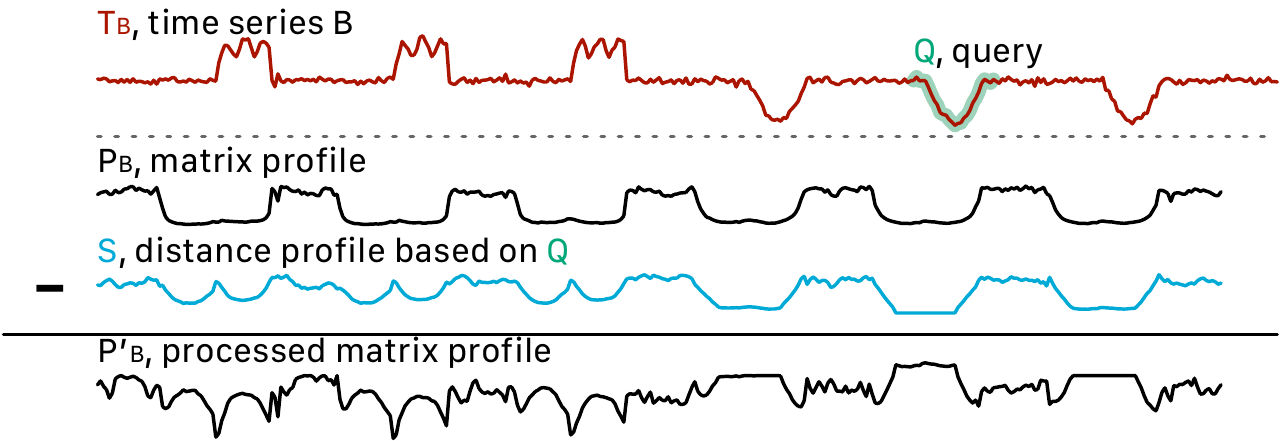}
}
\caption{
\footnotesize
The intra-similarity join matrix profile~$P_B$ is extracted from~\textcolor{time_series}{$T_B$}, and the distance profile~\textcolor{distance_pro}{$S$} is computed by querying~\textcolor{time_series}{$T_B$} with~\textcolor{query}{$Q$} (i.e., the subsequence added to the dictionary).
The processed matrix profile~$P'_B$ is computed by $P_B-\textcolor{distance_pro}{S}$.
$P'_B$ reveals repeated patterns dissimilar to~\textcolor{query}{$Q$}.
}
\label{fig:process_mp}
\end{figure}

We use Fig.~\ref{fig:process_mp} to visually demonstrate the importance of line 7 where the query~$Q$ is the subsequence added to the dictionary in the first iteration (i.e., $T_B[j:j+m]$ from the first iteration) and we are currently at line 7 in the second iteration.
If we extract the subsequence associated with the lowest value in~$P_B$, we may extract the \textit{V-shaped pattern} (\includegraphics[height=1.1em]{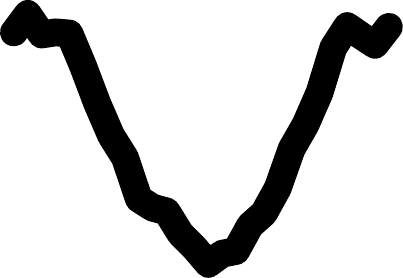}), which is similar to $Q$, instead of the novel \textit{crown-shaped pattern} (\includegraphics[height=1.1em]{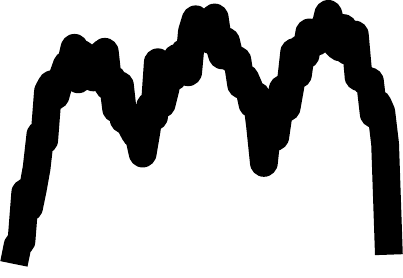}).
However, after~$S$ (from the first iteration) is subtracted from~$P_B$, the lowest value of the processed matrix profile~$P'_B$ will lead us to the crown-shaped pattern.

From line 8 to line 10, subsequences associated with the subsequences already added to the dictionary are removed from~$P'_B$ by replacing their corresponding matrix profile values with infinity.
The $\pm \frac{m}{2}$ are the exclusion zone for removing trivial matches~\cite{yeh2016matrix} of each $D_B \in \mathbf{D}_B$.
In line 11, the index associated with the best candidate for adding to the dictionary is identified.
As we already demoted the subsequences similar to the dictionary elements in line 7 and removed the dictionary elements in line 10, the minimum value in~$P'_B$ corresponds to a time series motif that has not yet been added to the dictionary.

\begin{figure}[htbp]
\centerline{
\includegraphics[width=0.99\linewidth]{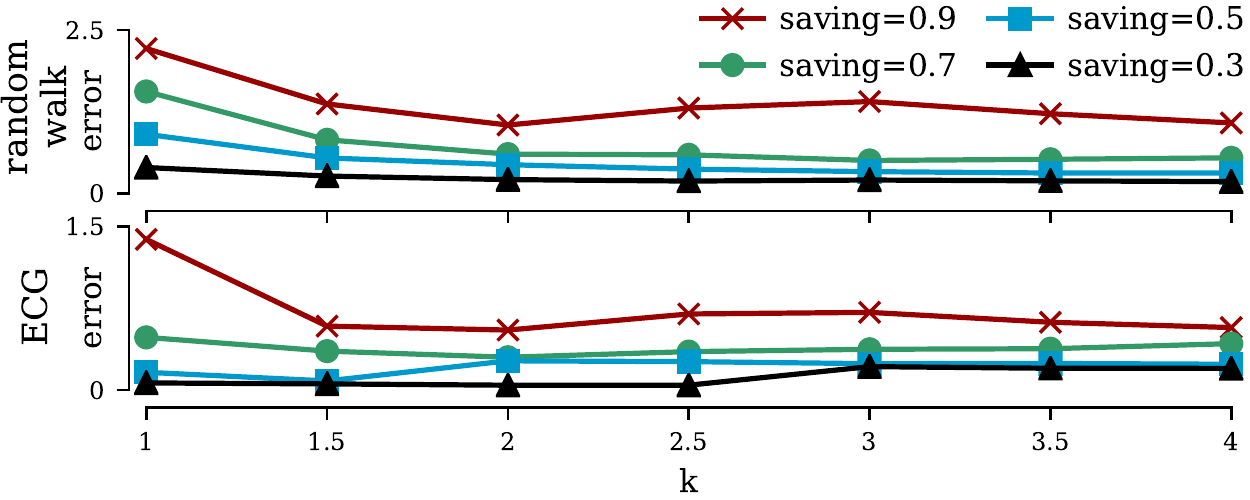}
}
\caption{
\footnotesize
The approximation error for different settings of context window factor~$k$ on random walk time series and heartbeat (electrocardiogram/ECG) time series under different space saving factors.
}
\label{fig:context_k}
\end{figure}

The discovered subsequence is added to the dictionary in line 13.
The user has the option to add additional data points preceding and following the selected subsequence as context to the dictionary.
A larger context window helps capturing different shifts of a pattern when~$T_B$ is highly periodic.
In our implementation, the length of the context window is parameterized by the hyper-parameter~$k$ where the length of context window is computed by $km$.
Fig.~\ref{fig:context_k} shows how~$k$ could affect the quality of the approximation.
We tested on both random walk time series and heartbeat time series under different space saving factors (i.e., $1 - \frac{\text{compressed size}}{\text{uncompressed size}}$).
We measure the quality of the dictionary constructed using different $k$ by computing the average Euclidean distance between the oracle matrix profile and the approximate matrix profile (i.e., average error) when joining the dictionary with $1,000$ random walk time series.
In general, setting~$k$ in $[1.5, 2]$ gives the most accurate approximate matrix profile, and the proposed method is \textit{not} sensitive to $k$ if $k$ is in $[1.5, 4]$.
For our experiments, we simply fixed $k$ to 1.5.

\definecolor{d0}{HTML}{990000}
\definecolor{d1}{HTML}{339966}
\definecolor{d01}{HTML}{FFCC00}

\begin{figure}[htbp]
\centerline{
\includegraphics[width=0.85\linewidth]{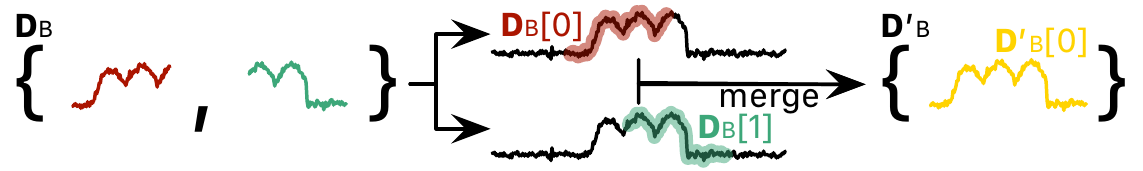}
}
\caption{
\footnotesize
Overlapping dictionary elements (i.e., \textcolor{d0}{$\mathbf{D}_B[0]$} and \textcolor{d1}{$\mathbf{D}_B[1]$}) are merged into one element~\textcolor{d01}{$\mathbf{D}'_B[0]$} for compact storage of dictionary elements.
}
\label{fig:dict_storage}
\end{figure}

To compactly store the elements of dictionary, we merged overlapping elements as shown in Fig.~\ref{fig:dict_storage}.
As shown in this figure, because dictionary element~$\mathbf{D}_B[0]$ overlapped with dictionary element~$\mathbf{D}_B[1]$ in the input time series, we store a longer subsequence (i.e., $\mathbf{D}'_B[0]$) covering both $\mathbf{D}_B[0]$ and $\mathbf{D}_B[1]$ so that the overlapping region is not stored twice.
We use $\mathbf{D}_B[i]$ to denote the $i$th element in~$\mathbf{D}_B$.

In line 13 and line 14, the terminal condition is checked.
The algorithm exits the loop and returns~$\mathbf{D}_B$ if the terminal condition is met.
The proposed algorithm supports two terminal conditions: 1) termination based on memory usage and 2) termination based on max error~$e_\text{max}$.
Users can choose the terminal condition based on their specific applications.
For example, if the application requires performing similarity join on a machine with limited memory, then the user may wish to terminate the dictionary learning algorithm with condition 1.
If the quality of the approximate solution is more critical for the application, then the user should choose to terminate based on condition 2 instead.

In line 15, the distance profile~$S'$ is computed by querying time series~$T_B$ with the dictionary element added in the current iteration.
We use the MASS algorithm~\cite{mass} to compute the distance profile in our implementation.
Then, the distance profile~$S'$ is merged with~$S$ using element-wise minimum operation from line 16 to line 19.
The resulting merged distance profile~$S$ measures the distance between each subsequence with each element in the dictionary, and is used in the next iteration in line 7 to process~$P_B$.
Finally, the learned dictionary~$\mathbf{D}_B$ is returned in line 20.

For our implementation, the time complexity is $O(max(n^2, n_\text{iter} n\log n))$ where $n$ is the length of $T_B$ and $n_\text{iter}$ is number of iterations which depends on the terminal condition.
The $n^2$ term is from the intra-similarity join operation in line 3 and the $n_\text{iter} n\log n$ is from calling MASS algorithm\footnote{MASS algorithm has a $O(n\log n)$ time complexity~\cite{mass}.} $n_\text{iter}$ times in line 15.
Note, because we avoid selecting trivial matches of any previously added dictionary element using an exclusion zone of length~$m$ in line 10, the maximum possible $n_\text{iter}$ are $\frac{n}{m}$.
Because $m$ is usually greater than $\log n$, this would make $n_\text{iter} n\log n$ less than $n^2$, and the overall time complexity becomes $O(n^2)$.
The time complexity can be further reduced if the join operation in line 3 is replaced with an approximate intra-similarity join algorithm.
Similar to other matrix profile algorithms~\cite{yeh2018time,zhu2018matrix,zhu2020swiss}, the space complexity is $O(n)$.

\subsection{Similarity Join with Dictionary.}
\label{sec:similarity}
Once the dictionary~$\mathbf{D}_B$ is learned from $T_B$, we can use Algorithm~\ref{alg:dict_join} to compute the approximate inter-similarity join matrix profile~$\hat{P}_{AB}$ for any given time series~$T_A$.
The inputs to the algorithm are time series~$T_A$, dictionary~$\mathbf{D}_B$, and subsequence length~$m$.

\begin{algorithm}[ht]
    \centering
    \caption{Similarity Join with Dictionary\label{alg:dict_join}}
    \footnotesize
    \begin{algorithmic}[1]
        \Input{time series~$T_A$, dictionary~$\mathbf{D}_B$, subsequence length~$m$}
        \Output{matrix profile~$\hat{P}_{AB}$}
        \Function{JoinDictionary}{$T_A, \mathbf{D}_B, m$}
        \State $n \gets \textsc{ GetArraySize}(T_A)$
        \State $\hat{P}_{AB} \gets \textsc{ GetInfArrayOfSize}(n - m + 1)$
        \For{\textbf{each} $D_B$ \textbf{in} $\mathbf{D}_B$}
        \State $P_{AD} \gets T_A\bowtie_{\theta_{\textit{1NN}}, m}D_B$
        \State $\hat{P}_{AB} \gets \textsc{ ElementwiseMin}(\hat{P}_{AB}, P_{AD})$
        \EndFor
        \State \Return $\hat{P}_{AB}$
        \EndFunction
    \end{algorithmic}
\end{algorithm}

In line 2, the time series length~$n$ is inferred from the time series~$T_A$.
In line 3, we initialize the approximate matrix profile~$\hat{P}_{AB}$ as a vector of INFs of length $n - m + 1$.
From line 4 to line 6, the for-loop processes each $D_B \in \mathbf{D}_B$ to generate the approximate matrix profile~$\hat{P}_{AB}$.
In line 5, the inter-similarity join matrix profile between~$T_A$ and~$D_B$ is computed and stored in $P_{AD}$.
We use the algorithm outlined in~\cite{zimmerman2019matrixxiv} to perform the inter-similarity join.
Next, in line 6, $P_{AD}$ is combined with~$\hat{P}_{AB}$ using element-wise minimum.
The resulting~$\hat{P}_{AB}$ is returned in line 7.

\definecolor{time_series}{HTML}{990000}
\definecolor{query}{HTML}{339966}
\definecolor{anomaly}{HTML}{FFCC00}

\begin{figure}[htbp]
\centerline{
\includegraphics[width=0.99\linewidth]{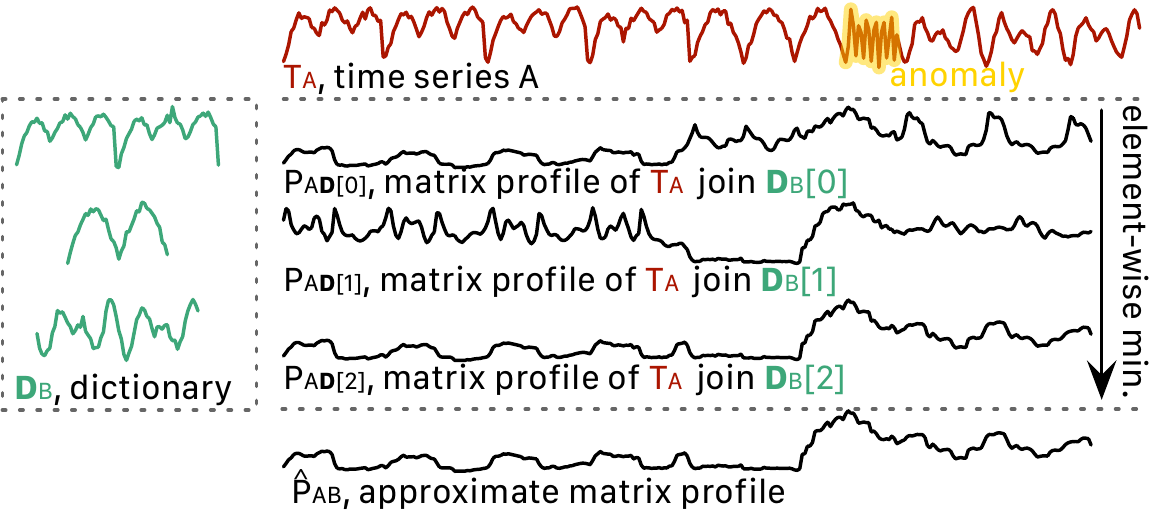}
}
\caption{
\footnotesize
The approximate matrix profile~$\hat{P}_{AB}$ is computed by combining the matrix profile between \textcolor{time_series}{$T_A$} and each element in \textcolor{query}{$\mathbf{D}_B$} with element-wise minimum.
The result~$\hat{P}_{AB}$ reveals an \textcolor{anomaly}{anomaly}.
}
\label{fig:dict_join}
\end{figure}

Fig.~\ref{fig:dict_join} provides a visual demonstration of Algorithm~\ref{alg:dict_join}.
In this example, the input time series~$T_A$ consists of three repeated patterns and one unique pattern (or \textit{anomaly}); the input dictionary~$\mathbf{D}_B$ happens to contain three elements, each corresponding to one of the three repeated patterns in~$T_A$.
Because each~$\mathbf{D}_B[i] \in \mathbf{D}_B$ consists of only one of the repeated patterns, each inter-similarity join matrix profile for each iteration (i.e., $P_{AD[i]}$) has low values associated with the corresponding repeated pattern.
By combining $P_{AD[0]}$, $P_{AD[1]}$, and $P_{AD[2]}$ with element-wise minimum, the result approximate matrix profile~$\hat{P}_{AB}$ reveals the anomaly.

The time complexity for our implementation is $O(|\mathbf{D}_B|n)$ where $|\mathbf{D}_B|$ is the number of data points in $\mathbf{D}_B$ and $n$ is the length of $T_A$.
Because the time complexity for line 5 is $O(|D_B|n)$ and $\sum_{D_B \in \mathbf{D}_B}{|D_B|} =|\mathbf{D}_B|$, the total time complexity is $O(|\mathbf{D}_B|n)$.

\subsection{Performance Guarantees.}
\label{sec:guarantee}
In this section, we demonstrate that the proposed method provides our two claimed assurances: 1) no false negative guarantee and 2) max error guarantee.
To help us use a more concise expression in the proofs, we further define the \textit{all-subsequences set} of a dictionary.

\begin{define}
    An \textit{all-subsequences set}~$\mathbf{D}_{B,m}$ of a dictionary~$\mathbf{D}_B$ is a set of all possible subsequences of $\mathbf{D}_B$ obtained by sliding a window of length~$m$ across each $D_B \in \mathbf{D}_B$, where $m$ is the subsequence length.
\end{define}

Note, this notation is created for explanation purpose only.
We do not actually compute such set in the proposed algorithms.
We also define the \textit{nearest neighbor distance function} for conciseness.

\begin{define}
The \textit{nearest neighbor distance function}~$\textsc{NNDist}(X, \mathbf{Y})$ is a function that finds the nearest neighbor of $X \in \mathbb{R}^{m}$ (a time series) in $\mathbf{Y}=\{Y_1, Y_2, ..., Y_n\}$ (a set of time series and $Y_i \in \mathbb{R}^{m}$); then return the distance between $X$ and the nearest neighbor.
\end{define}

\vspace{0.5em}
\noindent \textbf{[No False Negative Guarantee]}
To show that our method provides Guarantee~\ref{gua:false}, we need to prove Lemma~\ref{lem:false}.

\begin{lemma}
\label{lem:false}
For any given time series~$T_A$, if $P_{AB}$ and $\hat{P}_{AB}$ are defined using distance measure, then we have $\hat{P}_{AB}[i] \geq P_{AB}[i]$ for all $i$ and for any given~$T_A$.
\end{lemma}

Before we begin the proof, let us first establish Proposition~\ref{prop:false}.

\begin{proposition}
\label{prop:false}
Given a time series~$T_B$ and a subsequence length~$m$, the all-subsequences set~$\mathbf{D}_{B,m}$ of a dictionary~$\mathbf{D}_B$ returned by Algorithm~\ref{alg:dictionary} is always a subset of~$\mathbf{B}_m$ (i.e,. the all-subsequences set of~$T_B$ of length $m$).
\end{proposition}

Proposition~\ref{prop:false} is true because Algorithm~\ref{alg:dictionary} constructs dictionaries by extracting subsequences from~$T_B$.
We will use Proposition~\ref{prop:false} to prove Lemma~\ref{lem:false}.

\begin{proof}
Assume, for some $i$, we have $\hat{P}_{AB}[i] < P_{AB}[i]$, then there is a subsequence $U \in \mathbf{D}_{B,m}$, has a smaller distance than any subsequence in $\mathbf{B}_m$ to the $i$th subsequence $T_A[i]$ in $T_A$.
Therefore, $U \notin \mathbf{B}_m$.
Since $\mathbf{D}_{B,m} \subset \mathbf{B}_m$ (Proposition~\ref{prop:false}) and $U \in \mathbf{D}_{B,m}$, $U \in \mathbf{B}_m$.
This is a contradiction.
Thus, for all $i$, we have $\hat{P}_{AB}[i] > P_{AB}[i]$.
\end{proof}

As Lemma~\ref{lem:false} holds for all $i$ and any $T_A$, the proposed method is guaranteed to have the no false negatives.

\vspace{0.5em}
\noindent \textbf{[Max Error Guarantee]}
In order to show our method provides Guarantee~\ref{gua:error}, we need to first define $e_{\text{max}}$.

\begin{define}
Given a time series $T_B$, subsequence length~$m$, and a dictionary~$\mathbf{D}_B$ learned from~$T_B$, we have $e_{\text{max}} \triangleq \max_{T_{i,m} \in \mathbf{B}_m} \textsc{NNDist}(T_{i,m}, \mathbf{D}_{B,m})$.
\end{define}

With $e_{\text{max}}$ defined, we need to show Lemma~\ref{lem:error} holds.

\begin{lemma}
\label{lem:error}
Given a time series $T_B$, subsequence length~$m$, and a dictionary~$\mathbf{D}_B$ learned from~$T_B$, for any time series $T_A \in \mathbb{R}^{n}$, $D_{AD}[i] - D_{AB}[i] \leq e_{\text{max}}$ for all $i \in [1, n - m + 1]$ where $D_{AD}[i] = \textsc{NNDist}(T_A[i], \mathbf{D}_{B,m})$ and $D_{AB}[i] = \textsc{NNDist}(T_A[i], \mathbf{B}_m)$.
\end{lemma}

\begin{proof}
A subsequence of length $m$ can be viewed as a point in $\mathbb{R}^{m}$.
Since $\mathbf{D}_{B,m} \subset \mathbf{B}_m$, then every point in $\mathbf{D}_{B,m}$ coincides with one point in $\mathbf{B}_m$.
Then, for any possible $T_A[i] \in \mathbb{R}^{m}$, if the nearest neighbor of $T_A[i]$ in $\mathbf{B}_m$ is also in $\mathbf{D}_{B,m}$, then $D_{AD}[i] - D_{AB}[i] = 0$.
In another case, if the nearest neighbor of $T_A[i]$ in $\mathbf{B}_m$ is \textit{not} in $\mathbf{D}_{B,m}$ as shown in Fig.~\ref{fig:proof2}, suppose $T_A[i]$ is labeled as $a$, the nearest neighbor of $T_A[i] \in \mathbf{B}_m$ is labeled with $b$, and the nearest neighbor of $T_A[i]$ in $\mathbf{D}_{B,m}$ is labeled with $c$.
By triangle inequality $|\vec{ab}| - |\vec{ac}| < |\vec{bc}|$, and with $|\vec{bc}| \leq e_{\text{max}}$, thus $D_{AD}[i] - D_{AB}[i] \leq e_{\text{max}}$.
\end{proof}

\definecolor{d_m}{HTML}{990000}
\definecolor{b_m}{HTML}{339966}

\begin{figure}[htbp]
\centerline{
\includegraphics[width=0.99\linewidth]{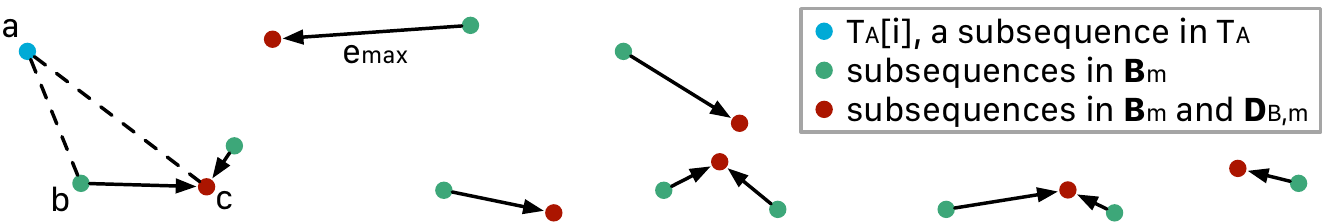}
}
\caption{
\footnotesize
Demonstration of the max error guarantee by triangle inequality.
The nearest neighbor of each \textcolor{b_m}{$\mathbf{B}_m[i] \in \mathbf{B}_m$} within the set \textcolor{d_m}{$\mathbf{D}_{b,m}$} is indicated with an arrow.
}
\label{fig:proof2}
\end{figure}

Because Lemma~\ref{lem:error} holds for any subsequence of length $m$, the approximation error of the proposed method is guaranteed to be bounded above by~$e_{\text{max}}$ for any given time series $T_A$.

\section{Empirical Evaluation }
\label{sec:experiment}
To ensure that our experiments are reproducible, we have built a website~\cite{supplementary} which contains all supplementary material for the results, in addition to some experiments that are omitted here for brevity.
Unless otherwise stated, all experiments were run with a 2.9 GHz Quad-Core Intel Core i7 processor (4 threads) and 16 GB ram.

\subsection{Sanity Check.}
To confirm the idea that different time series have different levels of ``compressibility" as hinted in Fig.~\ref{fig:compress}, here we learn dictionaries with Algorithm~\ref{alg:dictionary} using extended versions of the three time series shown in that figure, under different space saving factors.
The time series shown in Fig.~\ref{fig:compress} only contains the first thousand time steps, the full time series has eight thousand time steps.
We measure the max error~$e_\text{max}$ associated with each dictionary and summarize it in Fig.~\ref{fig:sanity_check}.

\begin{figure}[htbp]
\centerline{
\includegraphics[width=0.99\linewidth]{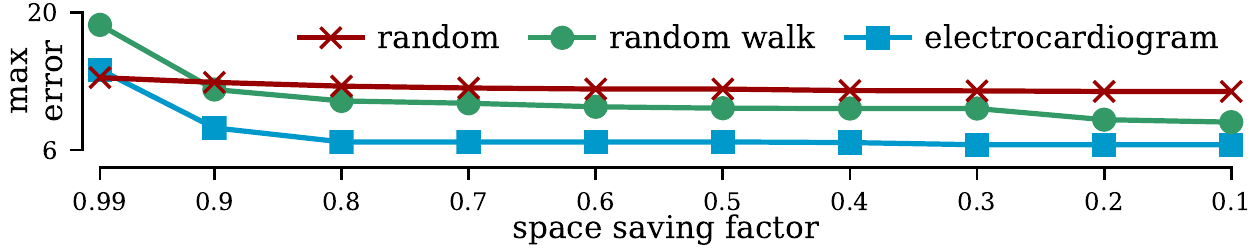}
}
\caption{
\footnotesize
The max error~$e_\text{max}$ associated with the most compressible time series (i.e., electrocardiogram) is the lowest of the thee when the space saving factor is smaller than or equal to 0.9.
}
\label{fig:sanity_check}
\end{figure}

First of all, the max error~$e_\text{max}$ is reduced as we are reducing the space saving factor.
Such correlation is expected as a larger dictionary is capable of capturing more information; therefore, the max error is lower.
Secondly, the max error~$e_\text{max}$ is essentially unchanged for random time series.
This is also expected as random time series does not contain any exploitable patterns.
The subsequences in a random time series do not correlate with each other, i.e., the Pearson correlation coefficient between them is always close to zero.
Third, the~$e_\text{max}$ associated with both random walk and ECG decreases as we increase the dictionary size.
The change is drastic for the lower space saving factors; then, the error stops improving when the space saving factor is greater than 0.8.
The ECG dictionary has the lowest~$e_\text{max}$ when the space saving is smaller than or equal to 0.9.

To empirically validate the derived theoretical max error~$e_\text{max}$, we join the dictionaries learned from ECG time series under different space saving factors with another much longer ECG time series of length 10 million.
Using the notation we presented in Section~\ref{sec:background} and Section~\ref{sec:algorithm}, the dictionaries~$\mathbf{D}_B$ are learned from the $8k$ length time series~$T_B$, and we join~$\mathbf{D}_B$ with the $10m$ length time series~$T_A$.
We compare the approximate matrix profile~$\hat{P}_{AB}$ with the exact matrix profile~$P_{AB}$ and compute the error vector~$E$ as the difference between the two matrix profiles, i.e., $\hat{P}_{AB} - P_{AB}$.
We summarize~$E$ using mean, standard deviation, maximum, and minimum as shown in Fig.~\ref{fig:maxerr_check}.

\begin{figure}[htbp]
\centerline{
\includegraphics[width=0.99\linewidth]{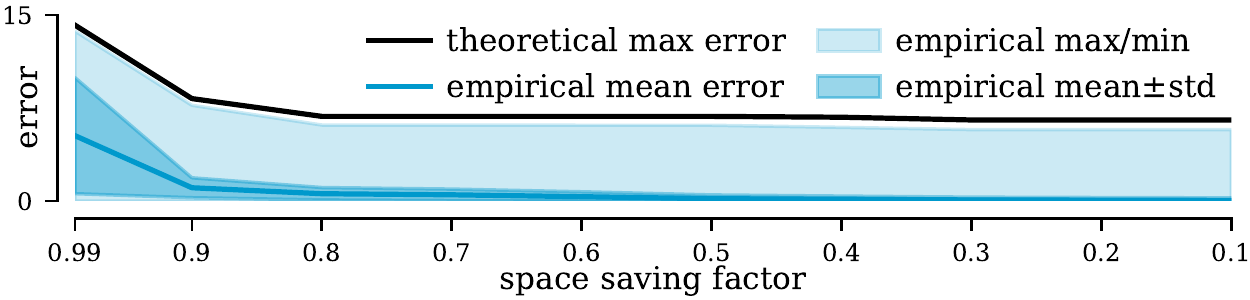}
}
\caption{
\footnotesize
The empirical error never exceeds the theoretical max error~$e_\text{max}$.
}
\label{fig:maxerr_check}
\end{figure}

The most important observation from the figure is that the empirically observed maximum error never exceeds the theoretical~$e_\text{max}$ (i.e., Guarantee~\ref{gua:error}).
In other words, in additional to the analytical result presented in Section~\ref{sec:guarantee}, we empirically validate the theoretical~$e_\text{max}$ is indeed the error bound.
Another observation is that the standard deviation reduced more drastically compared to the mean and the maximum error at lower space saving factors; the minimum error is almost always close to zero.
This suggests that for highly structured data like ECG, we can create an accurate approximation using only a small fraction of the original data.

To evaluate whether Algorithm~\ref{alg:dictionary} is capable of creating compact and representative dictionaries, we compare Algorithm~\ref{alg:dictionary} with the random sampling baseline\footnote{The random baseline randomly selects a subsequence from the time series and add the selected subsequence to the dictionary at each iteration.} using the 2017 Melbourne Pedestrian Dataset~\cite{melbournepedestrian}.
We follow the data processing step used in UCR Archive~\cite{ucrarchive2018} to process the original dataset into the UCR Archive format with 10 different classes.
Each class has 349 to 365 time series and each time series's length is 24.
In each trial, we randomly select one class as the background (majority) class and treat the other classes as the foreground (minority) class.
We concatenate all time series from the background class to form the background time series.
For each foreground class, we randomly select 2 to 16 time series using a random number generator, and insert the selected time series to the background time series.
Fig.~\ref{fig:accuracy_ts} depicts an example of the resulting time series.

\begin{figure}[htbp]
\centerline{
\includegraphics[width=0.99\linewidth]{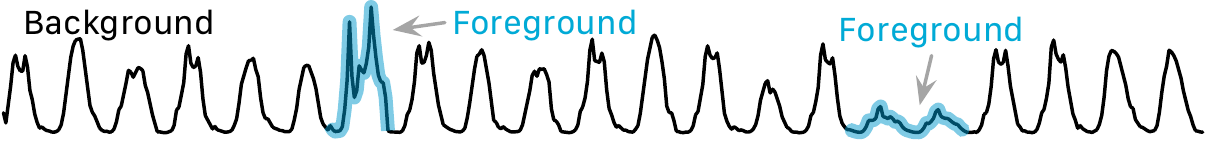}
}
\caption{
\footnotesize
Example time series extracted from Melbourne Pedestrian Dataset.
}
\label{fig:accuracy_ts}
\end{figure}

As both Algorithm~\ref{alg:dictionary} and the random baseline grow the dictionary incrementally, we run both methods until the dictionary captures time series from all 10 classes and record the associated space saving factor.
We use the computed space saving factors as a measure of a dictionary's quality.
We average over a thousand trials and summarize the results with the histogram shown in Fig.~\ref{fig:accuracy_test}.

\begin{figure}[htbp]
\centerline{
\includegraphics[width=0.99\linewidth]{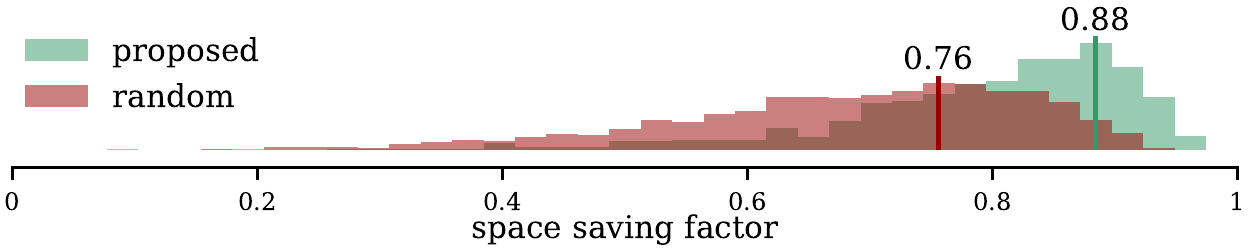}
}
\caption{
\footnotesize
The proposed method outperforms the random baseline.
}
\label{fig:accuracy_test}
\end{figure}

The proposed method generally has a much higher space saving factors (i.e., better compression) comparing to the random baseline.
The most populated bins for the proposed method and random baseline are 0.88 and 0.76, respectively.
Additionally, the proposed method performs more consistently as it has a more concentrated distribution comparing to the random baseline.

\subsection{Runtime.}
We presented two algorithms in Section~\ref{sec:algorithm}: dictionary building algorithm (i.e., Algorithm~\ref{alg:dictionary}) and dictionary join algorithm (i.e., Algorithm~\ref{alg:dict_join}).
Here, we test the runtime for both algorithms under different settings.
For all the experiments presented in this section, we set the number of threads to 4, we test the algorithms on the same ECG time series introduced in Fig.~\ref{fig:compress}.

\vspace{0.5em}
\noindent \textbf{[Dictionary Building Runtime]}
Based on our time complexity analysis, the runtime of Algorithm~\ref{alg:dictionary} is dependent on the length of input time series (i.e., $n$) and the number of iterations (i.e., $n_\text{iter}$).
We can easily change the length of input time series by using the first $n$ points as the input time series.
To vary $n_\text{iter}$, we use different terminal conditions based on the space saving factor.
We fix subsequence length~$m$ and contextual window factor~$k$ to 100 and 1.5, respectively.
The result runtime is presented in Fig.~\ref{fig:build_runtime}.

\begin{figure}[htbp]
\centerline{
\includegraphics[width=0.99\linewidth]{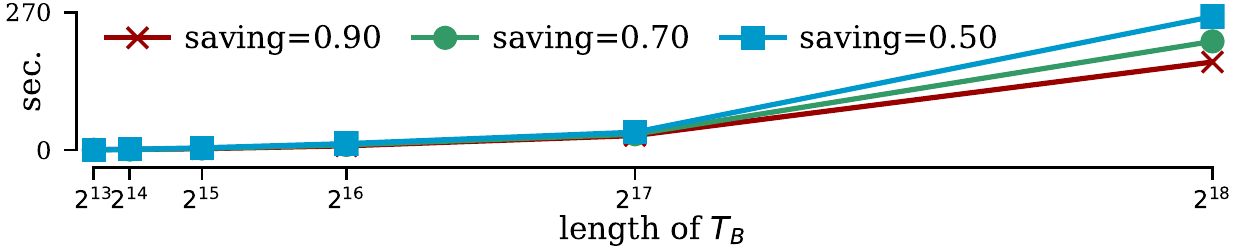}
}
\caption{
\footnotesize
Runtime for building dictionaries under different $T_B$ lengths/space saving factors.
}
\label{fig:build_runtime}
\end{figure}

Unsurprisingly, the runtime grows quadratically with the length of time series~$T_B$, which agrees with our analysis.
The runtime also decreases when the space saving factor increases as less iterations are required to construct a smaller dictionary.
To examine the relationship between dictionary building time and space saving factor closely, we set the length of $T_B$ to $2^{18}$ and measure the runtime on various space saving factors.
The result is presented in Fig.~\ref{fig:build_runtime_iter}.

\begin{figure}[htbp]
\centerline{
\includegraphics[width=0.99\linewidth]{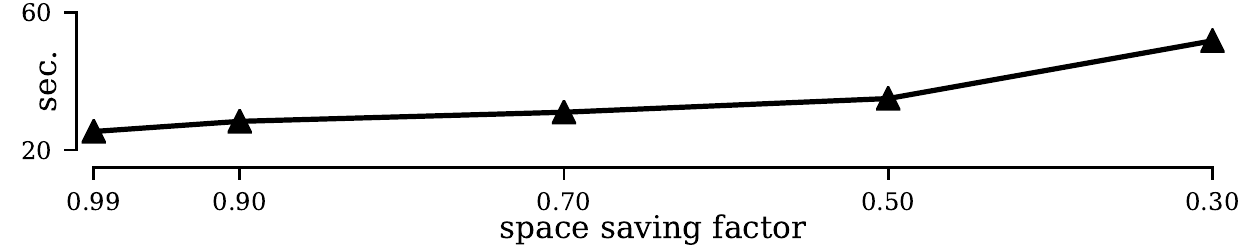}
}
\caption{
\footnotesize
Runtime of building dictionaries with different space saving factors.
}
\label{fig:build_runtime_iter}
\end{figure}

The runtime for building dictionaries generally increases as the desired dictionary is larger (i.e., smaller space saving factor).
From Fig.~\ref{fig:build_runtime_iter}, we notice that the runtime difference below space saving factor of 0.5 is much larger comparing to the runtime difference above 0.5.
In other words, the reduction on runtime is much noticeable for larger space saving factors comparing to smaller space saving factors.

Next, we evaluate the effect of varying subsequence lengths on runtime.
For this set of experiments, we set the length of $T_B$ to $2^{18}$ and space saving factor to 0.5 and the result is shown in Fig.~\ref{fig:build_runtime_m}.

\begin{figure}[htbp]
\centerline{
\includegraphics[width=0.99\linewidth]{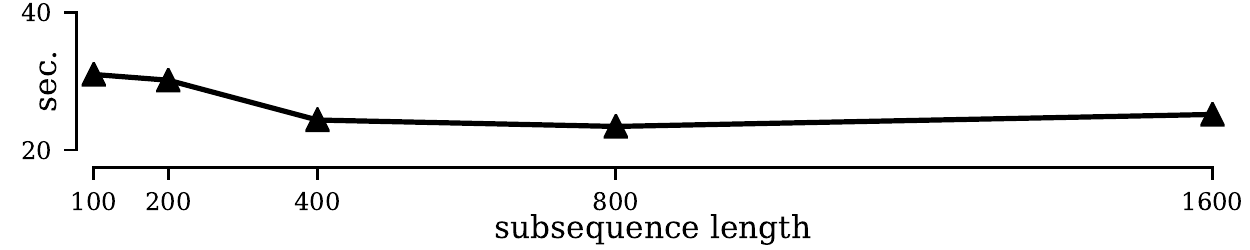}
}
\caption{
\footnotesize
Runtime of building dictionaries with different subsequence lengths.
}
\label{fig:build_runtime_m}
\end{figure}

The dictionary runtime reduces as the subsequence length increases.
Such phenomenon occurs because increased subsequence length actually reduces the number of subsequences.
This observation is consistent with other matrix profile algorithms~\cite{yeh2018time,zhu2018matrix,zhu2020swiss} where increasing dimensionality (i.e., subsequence length) actually is (very slightly) favorable in terms of the runtime.

\vspace{0.5em}
\noindent \textbf{[Dictionary Join Runtime]}
Similar to Algorithm~\ref{alg:dictionary}, the length of time series impacts the runtime noticeably.
In this case, there are two time series: the time series where the learned dictionary~$\mathbf{D}_B$ is built from (i.e., $T_B$) and the time series that~$\mathbf{D}_B$ is joined with (i.e., $T_A$).
The other factor that also affects the runtime is the size of the dictionary (i.e., $|\mathbf{D}_B|$).
To control the time series lengths, we fix one at $2^{17}$ while varying the length of the other time series.
To vary the dictionary size, we once again use the space saving factor as the terminal condition.
Fig.~\ref{fig:join_runtime} summarizes the resulting runtime.
We also include the runtime of exact inter-similarity join algorithm to demonstrate the runtime benefit offered by the proposed method.

\begin{figure}[htbp]
\centerline{
\includegraphics[width=0.99\linewidth]{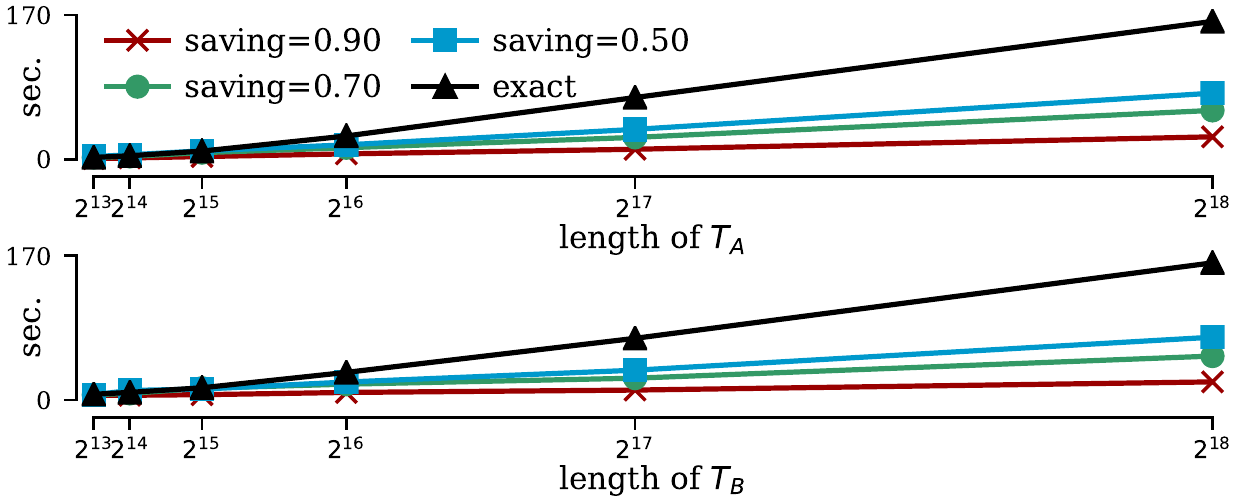}
}
\caption{
\footnotesize
Runtime for joining dictionaries with another time series under different time series lengths/space saving factors.
}
\label{fig:join_runtime}
\end{figure}

The reduction of runtime using the proposed approximate inter-similarity join algorithm is significant.
By setting the space saving factor to 0.5, the runtime is also reduced to roughly 50\% of the exact join in both line charts in Fig.~\ref{fig:join_runtime}.
Note, because the space saving factor for each line in the bottom chart of Fig.~\ref{fig:join_runtime} is fixed, changing the length of~$T_B$ directly varies the size of dictionary (i.e., $|\mathbf{D}_B|$).
On top of that, the length of~$T_A$ is $n$ in this case.
Given that the time complexity for Algorithm~\ref{alg:dict_join} is $O(|\mathbf{D}_B|n)$, the lengths of both~$T_B$ and~$T_A$ have a linear relationship with the runtime based on our big-O analysis, and the experiment result confirms such relationship.

We further examine the relationship between the space saving factor and runtime.
We set the lengths of both $T_B$ and $T_A$ to $2^{17}$ and measure the runtime on more space saving factors.
The result is presented in Fig.~\ref{fig:join_runtime_iter}.

\begin{figure}[htbp]
\centerline{
\includegraphics[width=0.99\linewidth]{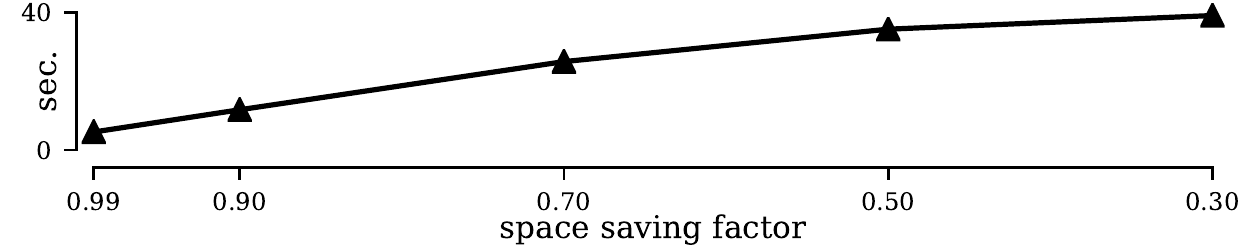}
}
\caption{
\footnotesize
Runtime for joining dictionaries with another time series under different space saving factors.
}
\label{fig:join_runtime_iter}
\end{figure}

The overall trend is very similar to Fig.~\ref{fig:build_runtime_iter}.
The key difference between Fig.~\ref{fig:join_runtime_iter} and Fig.~\ref{fig:build_runtime_iter} is the time difference in terms of percentage.
The runtime percentage reduction by increasing the space saving factor from 0.3 to 0.99 is 87\% for Fig.~\ref{fig:join_runtime_iter} while the percentage is 51\% for Fig.~\ref{fig:build_runtime_iter}.
This difference exists because the join operation in line 3 of Algorithm~\ref{alg:dictionary} does not depend on the setting of space saving factor while the loop from line 6 through line 20 depends on the setting of space saving factor.
In other words, Algorithm~\ref{alg:dictionary} always needs to pay a ``fixed" computational cost independent of the reduction in dictionary size.
While the entirety of Algorithm~\ref{alg:dict_join} benefits from smaller dictionaries in terms of runtime.

Finally, we measure the runtime of Algorithm~\ref{alg:dict_join} under different subsequence length settings.
We set the length of both $T_B$ and $T_A$ to $2^{17}$ and space saving factor to 0.5.
The result is presented in Fig.~\ref{fig:join_runtime_m}.

\begin{figure}[htbp]
\centerline{
\includegraphics[width=0.99\linewidth]{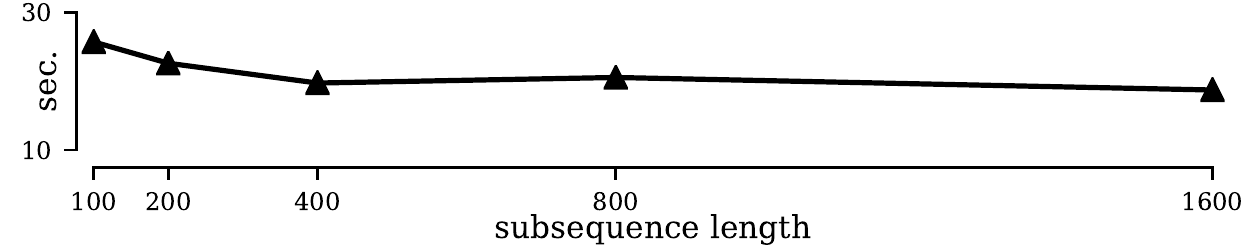}
}
\caption{
\footnotesize
Runtime for joining dictionaries with another time series under different subsequence lengths.
}
\label{fig:join_runtime_m}
\end{figure}

We can observe similar phenomena as in Fig.~\ref{fig:build_runtime_m} since both Algorithm~\ref{alg:dictionary} and Algorithm~\ref{alg:dict_join} use the same exact similarity join algorithm~\cite{zimmerman2019matrixxiv} as their subroutine.
The runtime generally decreases as the subsequence length increases.

\subsection{Case Study: Anomaly Detection.}
We evaluate the proposed method's anomaly detection capability with electrocardiogram (ECG) time series from MIT-BIH Long-Term ECG Database~\cite{goldberger2000physiobank}.

To be clear, the effectiveness of using time series discords has been established in many other works~\cite{yeh2018time,nakamura2020merlin}, here we are mostly interested in the implications of our work for efficiency.

For each patient, we use the first 1 million data points as the training time series, and the rest as the test time series.
There are seven patients in total, and the length of the test time series varies across patients from 5 million to almost 10 million.
We process the training time series by removing all the abnormal heartbeats (these have been annotated as such by a combination of algorithms and human expert annotation).
Dictionaries are learned from the processed training time series under different space saving factors; the resulting dictionaries are then joined with the test time series.
For the baseline of zero space saving factor, the processed training time series is directly joined with the test time series.

We use the inter-similarity join matrix profile value as the anomaly score for each subsequence, and we compute the \textit{area under the receiver operating characteristic curve} (AUC) to measure the quality of the anomaly score when comparing to the ground truth labels.
We compute the AUC and throughput under different space saving factors.
To summarize the results from seven patients, we average the AUC percentage change and throughput relative to the baseline (i.e., zero space saving factor) from each patient.
Note, because both the AUC percentage change and the throughput are computed relative to the baseline, the AUC percentage change and throughput corresponding to the baseline are 0\% and 1, respectively.
The experiment result is shown in Fig.~\ref{fig:anomaly}.

\begin{figure}[htbp]
\centerline{
\includegraphics[width=0.99\linewidth]{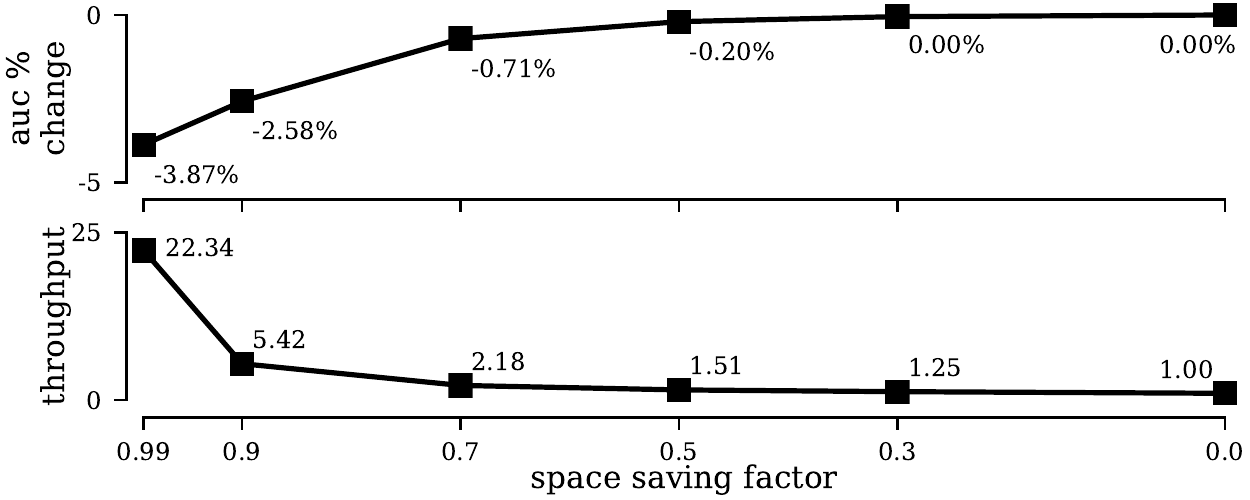}
}
\caption{
\footnotesize
When the space saving factor is set to 0.99, the percentage drop in AUC is less than 4\% and the the throughput is improved by at least 20X comparing to the baseline exact similarity join.
}
\label{fig:anomaly}
\end{figure}

It is unsurprising that the AUC decreases when we increase the space saving factor as we are using less and less resources; however, even with 99\% of the data removed, the AUC only drops 3.87\% on average.
When we examine the corresponding average throughput, it suggests that the approximate join algorithm is able to process 22.34 times more data compared to performing an exact join.
To put this speedup into context, exact join takes 30 minutes to process a million data points using the machine we performed the experiment.
If we replace the exact join with the proposed approximate join with a 0.99 space saving, the new system only takes 1.3 minutes to process the same data.
In other words, the exact join method can only handle a maximum sampling rate of approximately 500 Hz, while the proposed method is capable of handling 10,000 Hz (or twenty 500 Hz time series simultaneously) on the same machine.
Not to mention that the reduction in memory usage provides users the possibility of deploying the method on low-cost embedded systems.

Next, we take an excerpt of the heartbeat time series from one of the patients and visualize the output of our algorithm (with 0.99 space saving factor dictionary) along with the time series.
In this heartbeat sequence, there are two premature ventricular contraction (PVC/abnormal) heartbeats, labeled with ``V", in addition to the normal heartbeats.
Fig.~\ref{fig:anomaly_ex} shows the heartbeat time series.

\begin{figure}[htbp]
\centerline{
\includegraphics[width=0.99\linewidth]{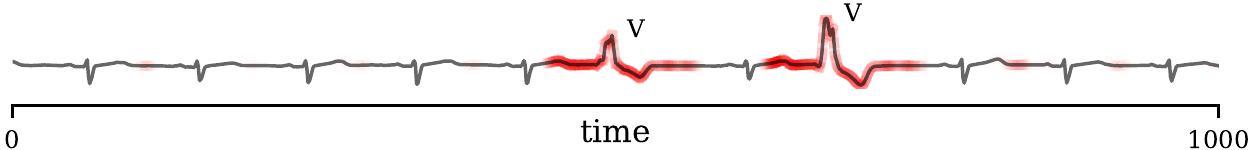}
}
\caption{
\footnotesize
The opacity of the \textcolor{red}{highlight} corresponds to the associated anomaly score (i.e., matrix profile value) at different temporal locations.
}
\label{fig:anomaly_ex}
\end{figure}

We highlight the time series based on the matrix profile outputted by the approximate algorithm where the opacity of the highlight is proportional to the corresponding matrix profile value at different temporal locations.
From Fig.~\ref{fig:anomaly_ex}, we can observe that the two PVC heartbeats are associated with high matrix profile values.
In other words, the proposed method, despite only using 1\% of the training data, still provides high-quality discriminative anomaly scores.

Finally, we compare similarity join-based method with deep learning-based methods.
We have implemented autoencoder with four popular architectures: multi-layer perceptron (MLP), convolutional neural network (CNN), recurrent neural network with long short-term memory (LSTM), and recurrent neural network with gated recurrent unit (GRU)~\cite{chalapathy2019deep}.
The implementation details regarding these methods can be found in the supporting website~\cite{supplementary}.
To have a fair comparison, we test our similarity join-based method using 4 different dictionary size settings: 0.18 MB, 0.38 MB, 0.51 MB, and 0.64 MB, matching the size of the four deep learning-based methods: MLP, CNN, GRU, and LSTM, respectively.
The resulting critical difference (CD) diagram comparing different methods over the 7 patients is shown in Fig.~\ref{fig:anomaly_cd}.

\begin{figure}[htbp]
\centerline{
\includegraphics[width=0.75\linewidth]{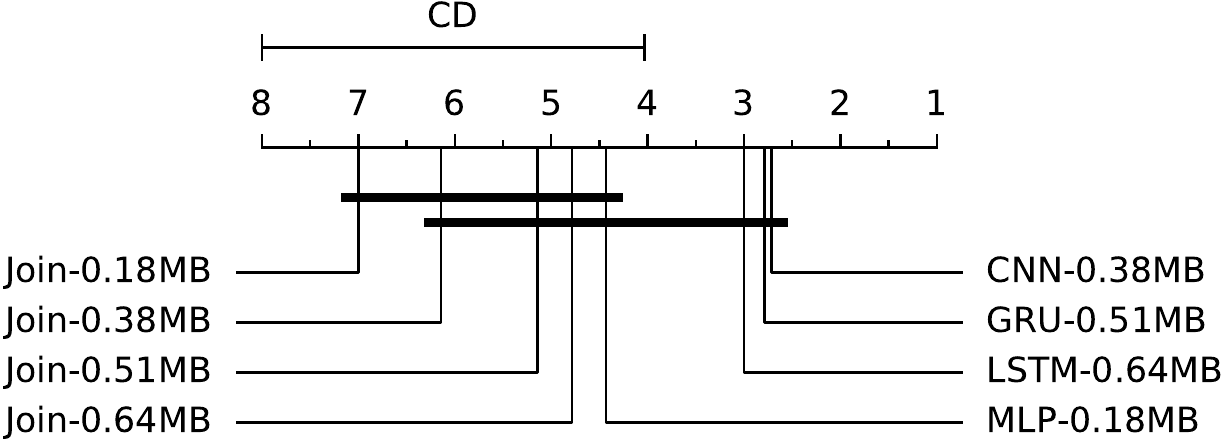}
}
\caption{
\footnotesize
The average rank (based on AUC) over the 7 patients for each method is marked on the horizontal axis.
The horizontal solid bars beneath the axis represent cliques, i.e. groups of methods within which there is no significant difference as determined by the pairwise Nemenyi test.
The performance difference between similarity join-based methods and deep learning-based methods is not significant when the model (dictionary) size is similar.
}
\label{fig:anomaly_cd}
\end{figure}

Overall, the differences among the tested methods are not significant when comparing methods with similar model (dictionary) sizes.
Deep learning-based methods generally have higher average rankings but do not provide the same performance guarantees as the proposed method.

\subsection{Case Study: Summarization.}
\label{sec:summary}
Thus far we have discussed our dictionary construction algorithm solely in terms of a technique to accelerate joins.
However, the readers will appreciate that the contents of the dictionary (and its metadata) may contain useful \textit{semantic} information about the dataset.
In a sense, the objective function optimized by the dictionary construction algorithm is similar to the objective function optimized by time series snippets~\cite{imani2018matrix}, and to a lesser extent, by clustering algorithms.

\definecolor{bundle_branch}{HTML}{F67F10}
\definecolor{premature_ventricular}{HTML}{2EA02C}
\definecolor{normal_beat}{HTML}{D62728}
\definecolor{tachycardia_beat}{HTML}{9467BD}

More explicitly, a side effect of creating a dictionary is that optimizing the dictionary's compression rate tends to summarize the data, by selecting one example of each unique behavior in the data.
To see this, we consider a four-minute snippet of an electrocardiogram (ECG)~\cite{greenwald1986development}.
The data was originally annotated with ``\textit{For the first 125 seconds the patient has \textcolor{bundle_branch}{Bundle branch block beats}, then there is a burst of \textcolor{premature_ventricular}{Premature ventricular contractions} lasting for about 28 seconds before the patient settles back to a mixture of \textcolor{normal_beat}{Normal} and \textcolor{tachycardia_beat}{Tachycardia} beats}."
In Fig.~\ref{fig:summary}, we show the result of building a dictionary on this dataset.

\begin{figure}[htbp]
\centerline{
\includegraphics[width=0.99\linewidth]{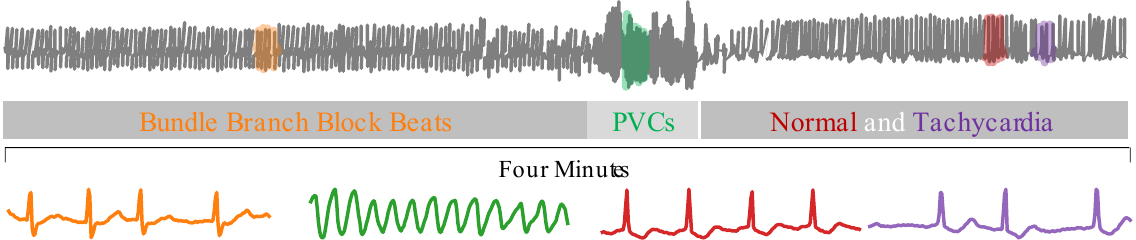}
}
\caption{
\footnotesize
\textit{top}) A four-minute 50 Hz ECG from a 75-year-old male undergoing cardiac surgery.
\textit{bottom}) The four elements of the dictionary with 12 to 1 compression.
}
\label{fig:summary}
\end{figure}

The four dictionary elements here exactly correspond to the four classes of beats that the original annotator had noted.
Moreover, this is a robust finding as we get essentially the same results even if we truncate some of the prefix/suffix, add noise etc.
\section{Conclusion}
\label{sec:conclusion}
In addition to the reduction in space/time complexity, the proposed approximate inter-time series similarity join method is guaranteed to 1) have no false negative and 2) have a proven max error.
The benefit of the method has been demonstrated with experiments in anomaly detection with heartbeat time series.
The proposed method has similar accuracy comparing to exact joins but the throughput is improved for more than 20X.
Also, the dictionary representation, that comes with the proposed method, provides a way to summarize long time series and allows the users to compare the dictionary representation in place of the original time series.
Our claims are verified with experiments on medical and transportation datasets.

\bibliographystyle{IEEEtran}

\end{document}